\documentclass[a4paper,onecolumn,11pt,unpublished]{quantumarticle}
\pdfoutput=1
\usepackage[utf8]{inputenc}
\usepackage[english]{babel}
\usepackage[T1]{fontenc}
\usepackage{amsmath}
\usepackage{hyperref}
\usepackage{graphicx}
\usepackage{tikz}
\usepackage{tikz-cd}
\usepackage{lipsum}
\usepackage{amssymb}
\usepackage{physics}
\usepackage{amsthm}
\usepackage{ytableau}
\usepackage{amsfonts}
\usepackage[numbers,sort&compress]{natbib}

\usepackage{array}
\def\id{{\rm 1\kern-.22em l}}
\usepackage{bm}

\newtheorem{teor}{Theorem}
\newtheorem{propo}{Proposition}

\begin{document}

\title{Unitary operator bases as universal averaging sets}

\author{Marcin Markiewicz}
\affiliation{Institute of Theoretical and Applied Informatics, Polish Academy of Sciences, ul. Baltycka 5, 44-100 Gliwice, Poland}
\orcid{0000-0002-8983-9077}
\author{Konrad Schlichtholz}
\affiliation{International Centre for Theory of Quantum Technologies (ICTQT),
University of Gdansk, 80-309 Gdansk, Poland}
\orcid{0000-0001-8094-7373}

\maketitle

\begin{abstract}
We provide a generalization of the idea of unitary designs to cover finite averaging over much more general operations on quantum states. Namely, we construct finite averaging sets for averaging quantum states over arbitrary reductive Lie groups, on condition that the averaging is performed uniformly over the compact component of the group. Our construction comprises probabilistic mixtures of unitary $1$-designs on specific operator subspaces. Provided construction is very general, competitive in the size of the averaging set when compared to other known constructions, and can be efficiently implemented in the quantum circuit model of computation.
\end{abstract}

\section{Introduction}

\subsection{General remarks about finite averaging}

Finite averaging of physical states is a task in which a physical state of some system is averaged over a representation of some continuous symmetry group, and the averaging process is represented as finite. This means that such a procedure can be operationally realized by sampling the representatives of symmetry operations from a \textit{finite averaging set}.  The original mathematical motivation for finite averaging in physics came from combinatorial designs, especially \textit{spherical designs} \cite{Delsarte77, Hong82}, in which the averages of polynomial functions on multidimensional spheres were realized by sampling values of the functions on finite subsets of points on these spheres. 
This idea was finally generalized to the averaging of polynomial functions defined on arbitrary manifolds equipped with finite measures \cite{Seymour84}, providing the most general version of the mean value theorem known from elementary calculus. 

In the context of quantum mechanics the idea of finite averaging needed some technical modifications related with the fact that quantum states are defined on complex manifolds; however, the core concept was the same: we average polynomial functions applied to complex vectors or matrices, and express these averages as finite sums of the respective polynomials evaluated on some finite sets of points. Historically, the first concepts of finite averaging sets in the quantum context were pure state designs \cite{Ambainis07} (known as projective $t$-designs,   inspired by the concept of designs on complex projective spaces \cite{Hoggar82}) and unitary $t$-designs \cite{Dankert05}, defined for averaging of quantum states over tensor products of unitary operations. Projective and unitary designs found several important applications within quantum information science, including efficient protocols for quantum state tomography \cite{Scott08, Bae19}, and analysis of random quantum circuits \cite{Brandao13, Hunter19}.  

In addition, several generalizations of these two concepts arose, including mixed-state designs \cite{Czartowski20}, special linear designs \cite{Markiewicz21}, or channel designs \cite{Czartowski25}. All these concepts share similar features. Firstly, they are guaranteed to exist for all dimensions; however, only particular examples are constructively known. Another feature of all of these finite averaging sets is that they perfectly preserve the structure of averaging operations. In this work, we propose an alternative approach to finite averaging of quantum states with respect to arbitrary symmetry groups. We trade the assumption of exact preservation of the structure of the symmetry operation for constructive character and competitive size of the finite averaging set. In order to introduce our concept, we need to define some preliminary ideas related to averaging and group representations.

\subsection{Averaging quantum states over symmetry groups}

A general blueprint for averaging quantum states over symmetry operations is provided by the so-called \textit{twirling map} \cite{Bartlett07}, which can be expressed as follows:
\begin{equation}
    \label{def-twirl-gen}
    \mathcal T_{\mathcal G}(\rho)=\int_{\mathcal G}\operatorname{d}\!g\,f(g)\pi_{\mathcal G}(g)\rho\,\pi_{\mathcal G}(g)^{\dagger},
\end{equation}
in which a finite-dimensional quantum system, represented by a density matrix $\rho$, is averaged over a symmetry group $\mathcal G$, which acts on quantum states via a representation $\pi_{\mathcal G}$, with respect to a measure $\operatorname{d}\!g$ and a weight function $f(g)$.
Interpretation of the above map is especially easy in the case of a compact symmetry group $\mathcal G$, since then one can always resort to unitary representations $\pi_{\mathcal G}$, which preserve normalization of the twirled state.
An emblematic example of such a situation is specified by taking $\pi_{\mathcal G}(U)=U^{\otimes t}$, namely the collective action of a unitary group on a system of $t$ particles.
In the case of non-compact symmetry groups, the physical meaning of the twirling map \eqref{def-twirl-gen} is not always well-defined, since appropriate quantum operation must be trace non-increasing. In the work \cite{Markiewicz23} a generalization of the twirling map \eqref{def-twirl-gen} to the case of arbitrary reductive Lie group has been proposed:
\begin{equation}
    \mathcal T_{\mathcal{G}}(\rho)=\int_{\mathcal K\times\mathcal A\times\mathcal K} \left(KA_{\textrm{n}}K'\right)\rho \left(KA_{\textrm{n}}K'\right)^{\dagger} \operatorname{d}\!K\operatorname{d}\!A \operatorname{d}\!K'.
    \label{mainTwirlG00}
\end{equation}
The above generalization is based on an iterated integral over the Cartan 'KAK' decomposition defined \textit{with respect to a given representation} of the group $\mathcal G$. The Cartan decomposition \cite{Knapp}, which is a generalization of the singular value decomposition (SVD) of matrices, allows one to represent any element $G$ of a given reductive Lie group $\mathcal G$ as a product of three elements: $K,A,K'$, in which $K$ and $K'$ are elements of a maximal compact subgroup $\mathcal K$ of $\mathcal G$, whereas $A$ belongs to a maximal Abelian subgroup $\mathcal A$.  

The crucial element of the above generalization is the normalization of  the operators representing the non-compact component of the group under consideration, $A_{\textrm{n}}=A/||A||$, which ensures trace non-increasing character of the map \eqref{mainTwirlG00}. Throughout this work we will utilize the definition of a generalized twirling in the form of \eqref{mainTwirlG00}, with the assumption of uniform averaging over compact components of the group $\mathcal{G}$, which basically means that $\int_{\mathcal K}  \operatorname{d}\!K$ represents the Haar integral over $\mathcal K$. In the work \cite{Markiewicz21} the above map has been used for $\mathcal G=\textrm{SL}(2,\mathbb C)$ for which $\mathcal K$ is the special unitary group $\operatorname{SU}(2)$ and $\mathcal A$ is isomorphic with $\mathbb R^+$.

In order to justify that the definition \eqref{mainTwirlG00} is sufficiently general, let us take an arbitrary, possibly reducible, \textit{matrix} representation $\pi_{\mathcal G}$ of the group 
$\mathcal G$ on $\mathbb C^n$. Let us denote by $\mathfrak c: \mathcal K \times \mathcal A \times \mathcal K \rightarrow  \mathcal G$ the mapping from the Cartesian product of the Cartan components of the group to the group itself:
\begin{equation}
    \label{def:C}
    \mathfrak c(K,A,K')=KAK',
\end{equation}
and by $\mathfrak d: \pi_{\mathcal G}(\mathcal K) \times \pi_{\mathcal G}(\mathcal A) \times \pi_{\mathcal G}(\mathcal K) \rightarrow  \pi_{\mathcal G}(\mathcal G)$ the corresponding map on the level of representations:
\begin{equation}
    \label{def:D}
    \mathfrak d(\pi_{\mathcal G}(K),\pi_{\mathcal G}(A),\pi_{\mathcal G}(K'))=\pi_{\mathcal G}(K)\pi_{\mathcal G}(A)\pi_{\mathcal G}(K').
\end{equation}
Consider now the following diagram:
$$
\begin{tikzcd}
\mathcal K \times \mathcal A \times \mathcal K \arrow{r}{\mathfrak c} \arrow[swap]{d}{\pi_{\mathcal G}} & \mathcal G \arrow{d}{\pi_{\mathcal G}} \\%
\pi_{\mathcal G}(\mathcal K)\! \times\! \pi_{\mathcal G}(\mathcal A) \!\times\! \pi_{\mathcal G}(\mathcal K) \arrow{r}{\mathfrak d}& \pi_{\mathcal G}(\mathcal G)
\end{tikzcd}
$$
The above diagram clearly commutes. To see this take an arbitrary triple $(K,A,K')\in\mathcal K\times\mathcal A\times\mathcal K$ and note that:
\begin{eqnarray}
    \label{diag:comm}
    &&({\mathfrak d}\circ\pi_{\mathcal G})(K,A,K')={\mathfrak d}(\pi_{\mathcal G}(K),\pi_{\mathcal G}(A),\pi_{\mathcal G}(K'))=\pi_{\mathcal G}(K)\pi_{\mathcal G}(A)\pi_{\mathcal G}(K'),\nonumber\\
    &&(\pi_{\mathcal G}\circ{\mathfrak c})(K,A,K')=\pi_{\mathcal G}(KAK')=\pi_{\mathcal G}(K)\pi_{\mathcal G}(A)\pi_{\mathcal G}(K'),
\end{eqnarray}
in which the last equality follows from the fact that any representation must be compatible with group multiplication.
This indicates that taking the Cartan-decomposition-based averaging is well defined for any choice of a representation $\pi_{\mathcal G}$ of $\mathcal G$.
Hence we can express the generalized twirling map \eqref{mainTwirlG00}
in a form which explicitly involves a given representation $\pi_{\mathcal G}$:
\begin{equation}
    \mathcal T_{\mathcal{G}}(\rho)=\int_{\mathcal K\times\mathcal A\times\mathcal K} \left(\pi_{\mathcal G}(K)\pi_{\mathcal G}(A_{\textrm{n}})\pi_{\mathcal G}(K')\right)\rho\left(\pi_{\mathcal G}(K)\pi_{\mathcal G}(A_{\textrm{n}})\pi_{\mathcal G}(K')\right)^{\dagger} \operatorname{d}\!K\operatorname{d}\!A \operatorname{d}\!K'.
    \label{mainTwirlG0}
\end{equation}

\subsection{Finite averaging}

Standard approach to formulating the finite averaging procedure over symmetry operations assumes that the twirling integral is expressed as a \textit{finite sum} with respect to some subset $\mathcal X$ of group elements called a \textit{finite averaging set}:
\begin{equation}
    \label{def-fin-group-average00}
    \mathcal T_{\mathcal G}(\rho)=\int_{\mathcal G}\operatorname{d}\!g\,f(g)\pi_{\mathcal G}(g)\rho\,\pi_{\mathcal G}(g)^{\dagger}=\sum_{g_i\in \mathcal X_{\mathcal G}} f(g_i)\pi_{\mathcal G}(g_i)\rho\,\pi_{\mathcal G}(g_i)^{\dagger}.
\end{equation}
An emblematic example of such finite averaging sets is the concept of unitary $t$-designs \cite{Dankert05, Gross07, Scott08, Dankert09, Roy09, Webb16, Hunter19}, for which we take $\pi_{\mathcal G}=U^{\otimes t}$ and $f(g)=1$. Such defined finite averaging sets always exist, see e.g. \cite{Markiewicz21}, following the seminal work \cite{Seymour84}, however in general they are very difficult to construct and only particular cases are already known even for small dimensions \cite{Gross07, Dankert09, Webb16}. 

In this work, we redefine the notion of finite averaging, by partially decoupling it from the group structure itself, and propose the following construction. The starting point is to introduce decomposition into irreducible representations, labeled by an integer index $k\in \mathcal I$:
\begin{equation}
    \label{def-fin-group-average0}
    \mathcal T_{\mathcal G}(\rho)=\int_{\mathcal G}\operatorname{d}\!g\,f(g)\pi_{\mathcal G}(g)\rho\,\pi_{\mathcal G}(g)^{\dagger}=\sum_{k\in\mathcal{I}} \int_{\mathcal G}\operatorname{d}\!g\,f(g)\pi_{\mathcal G}^{(k)}(g)\rho\,\pi_{\mathcal G}^{(k)}(g)^{\dagger},
\end{equation}
in which irreducible representations $\pi_{\mathcal G}^{(k)}$ of $\mathcal G$ are assumed to be trivially extended onto the entire representation space in order to avoid direct-sum notation of the form $\bigoplus_{k\in\mathcal I}\pi_{\mathcal G}^{(k)}$. The diagonal character of the above way of averaging (the fact that the $k$-th subspace on the "ket" side meets the $k$-th subspace on the "bra" side) follows from Schur's Lemma, which will be discussed thoroughly later in this work. Then by \textit{universal finite unitary averaging} we understand the following representation of the integral:
\begin{equation}
    \label{def-fin-group-average}
    \mathcal T_{\mathcal G}(\rho)=\int_{\mathcal G}\operatorname{d}\!g\,f(g)\pi_{\mathcal G}(g)\rho\,\pi_{\mathcal G}(g)^{\dagger}=\sum_{k\in\mathcal{I}}p_k \frac{1}{\left(D^k_G\right)^2} \sum_{l=1}^{\left(D^k_G\right)^2} \tilde\gamma_l^{(k)}\rho\,\tilde\gamma_l^{(k)\dagger}=\sum_{k\in\mathcal{I}}p_k\mathcal{T}_{\mathcal U}^k(\rho),
\end{equation}
in which the unitary operators $\tilde\gamma_l^{(k)}$ act on  $D^k_G$-dimensional subspaces irreducible under the action of the group $\mathcal G$ via representation $\pi_{\mathcal G}$, and are trivially extended on the entire representation space. The channel:
\begin{equation}
    \label{def-fin-group-average-mixed-unitary}
   \mathcal{T}_{\mathcal U}^k(\rho)=\frac{1}{\left(D^k_G\right)^2} \sum_{l=1}^{\left(D^k_G\right)^2} \tilde\gamma_l^{(k)}\rho\,\tilde\gamma_l^{(k)\dagger},
\end{equation}
is a uniformly mixed unitary channel on $k$-th invariant subspace implementing unitary $1$-design therein. In the case of averaging over a compact symmetry group we have $p_k=1$ for each $k$, whereas in the case of a noncompact one, $p_k$' s are generally less than one, and moreover $\sum_{k\in\mathcal I}p_k<1$, which indicates that averaging over such a symmetry group leads to a trace-decreasing operation. Physically, this means that the operation involves stochastic filtering of quantum states \cite{Avron09}, and some runs of the experiment have to be rejected; see, e.g. the case of averaging over Stochastic Local Operations and Classical Communication (SLOCC) \cite{Markiewicz21, Markiewicz23}.

In the following section, we will show that \textit{universal finite unitary averaging} exists for any reductive Lie group $\mathcal G$ acting via a regular finite-dimensional representation $\pi_{\mathcal G}$, whenever the averaging procedure is defined via Cartan decomposition \eqref{mainTwirlG0}, and the averaging over the compact components is performed in a uniform way. Our proof is constructive: we show that universal finite unitary averaging can be implemented using unitary operator bases \cite{BZ17}. The proposed finite averaging procedure allows to simulate a very broad class of averaging procedures for quantum states by simply sampling from a finite set of unitary operations. 

Note that the construction does not strictly follow the idea of finite group averaging in the original form \eqref{def-fin-group-average00}, as it does not reduce to sampling a finite number of elements of the original group and expressing them in the original representation.
On the other hand, it is constructive and universal and can be utilized as a means to analyze random quantum circuits \cite{Dankert05, Hunter19} and randomized quantum protocols \cite{Brandao13}.  

Despite being different than the original $t$-designs, the proposed procedure is not entirely decoupled from the group $\mathcal G$ structure, since this structure affects two aspects: (i) dimensions $D^k_G$ of irreducible subspaces, which determine the range of the index $l$ in formula \eqref{def-fin-group-average}, and (ii) explicit form of weights $p_k$, which are determined by integrating over the non-compact group component.

\section{Universal averaging with unitary operator bases}

\subsection{Generalized duality for group representations and Schur operator bases}

One of the main analytical tools in the theory of averaging over symmetry groups is  Schur-Weyl duality \cite{Goodman09, Dipper08, Doty09, Marvian14, Gross21}. However, in this work, we will utilize a more general approach to duality of group representations in order to provide the most general results, with Schur-Weyl duality being a specific example. 

Let us set the scene, and take some reductive subgroup $\mathcal G$ of a general linear group $\textrm{GL}(d, \mathbb C)$, which acts on a vector space $\mathbb C^{n}$ via a finite-dimensional regular representation $\pi_{\mathcal G}$. This representation is in general reducible, and natural question arises, how one can relate $\pi_{\mathcal G}$ with all irreducible representations $\{\xi_{\mathcal G}^{(k)}\}$ of the group $\mathcal G$, when acting on the vector space $\mathbb C^{n}$. The answer is specified by the \textit{general duality theorem} (see \cite{Goodman09}, sec. 4.2):

\begin{propo}[Duality theorem for group representations]
Let us take a reductive subgroup $\mathcal G$ of a general linear group $\operatorname{GL}(d, \mathbb C)$, its regular finite-dimensional representation 
$\pi_{\mathcal{G}}$  acting on $\mathbb C^n$, and the set of (equivalence classes of) its irreducible representations $\{\xi_{\mathcal G}^{(k)}\}$. Now let us introduce the centralizer (commutant) $\pi_{\mathcal C}$ of $\pi_{\mathcal{G}}$ in the algebra of endomorphisms of the vector space $\mathbb C^n$ and the corresponding set of (equivalence classes of) its irreducible representations $\{\xi_{\mathcal C}^{(k)}\}$. Then the joint action of $\pi_{\mathcal G}\times\pi_{\mathcal C}$ on $C^{n}$ decomposes uniquely into direct sum of products of irreducible representations $\xi_{\mathcal G}^{(k)}$, and $\xi_{\mathcal C}^{(k)}$:
\begin{equation}
\label{eq:isodecomp}
    \pi_{\mathcal G}\times\pi_{\mathcal C}=\bigoplus_{k\in\mathcal I} \xi_{\mathcal G}^{(k)}\otimes \xi_{\mathcal C}^{(k)}.
\end{equation}
\end{propo}
Note that each pair of irreducible representations appears only once in the above decomposition, which is the clue of the duality for group representations (the decomposition is \textit{multiplicity free}).
In order to make clear what the above \textit{virtual} tensor product means we introduce operator notation following \cite{Markiewicz23}.
Indeed, let us introduce the following basis on $\mathbb C^n$, called Schur basis, defined for each $k$-th subspace by the table:
\begin{eqnarray}
    \label{SchurBasis0}
     G^k_1:\,\,&&\ket{k,1,1}\,\,\,\,\,\,\,\,\,\ket{k,2,1}\,\,\,\,\,\,\ldots\,\,\,\,\,\,\ket{k,D^k_G,1} \nonumber\\
     G^k_2:\,\,&&\ket{k,1,2}\,\,\,\,\,\,\,\,\,\ket{k,2,2}\,\,\,\,\,\,\ldots\,\,\,\,\,\,\ket{k,D^k_G,2} \nonumber\\
     && \ldots\ldots\ldots\ldots\ldots\ldots\ldots\ldots\ldots\ldots\ldots\ldots \nonumber\\
       G^k_{D^k_C}:\,\,&&\underbrace{\ket{k,1,D^k_C}}_{C^k_1}\,\,\underbrace{\ket{k,2,D^k_C}}_{C^k_2}\,\,\,\ldots\,\,\,\underbrace{\ket{k,D^k_G,D^k_C}}_{C^k_{D^k_G}}.\nonumber\\
    \end{eqnarray}
The above table shows the interplay between the action of $\pi_{\mathcal G}$ and $\pi_{\mathcal C}$ on the irreducible subspaces for the joint action of $\pi_{\mathcal G}\times\pi_{\mathcal C}$. Namely, subspaces $G^k_{\lambda}$, spanned by vectors  $\{\ket{k,m,\lambda}\}_{m=1}^{D^k_G}$ occurring in each row, are invariant and irreducible under the action of $\pi_{\mathcal G}$, and correspond to equivalent irreducible representations $\xi^{(k)}_{\mathcal G}$ of dimension $D^k_G$.
On the other hand the subspaces $C^k_m$ spanned by vectors $\{\ket{k,m,\lambda}\}_{\lambda=1}^{D^k_C}$ appearing in each column are invariant and irreducible under the action of  $\pi_{\mathcal C}$, and correspond to equivalent $D^k_C$-dimensional irreducible representations $\xi^{(k)}_{\mathcal C}$. Now the table can be seen as representing tensor product of \textit{virtual vectors} $\ket{m_k}\otimes\ket{\lambda_{k}}$, which separately span virtual subspaces of dimensions $D^k_G$ and $D^k_C$ respectively. Although the picture of virtual subspaces is fruitful in some applications, in the context of the theory of group averaging it is rather more convenient to introduce  the following outer-product-based operators, as shown in recent works \cite{Markiewicz23, Schlichtholz24}:
    \begin{equation}
    \label{FullPiBasis}
    \hat\Pi_{kk'}^{m_1\lambda_1 m_2\lambda_2}=\ket{k,m_1,\lambda_1}\bra{k',m_2,\lambda_2}.
\end{equation}
With the help of the above operators we can define two \textit{Schur operator bases}:
\begin{eqnarray}
\label{PiBasis0}
\hat\Lambda^{\lambda_1\lambda_2}_k&=&\sum_{m=1}^{D^k_G}\hat\Pi_{kk}^{m\lambda_1 m\lambda_2},\nonumber\\
\hat\Pi^{m_1 m_2}_k&=&\sum_{\lambda=1}^{D^k_C}\hat\Pi_{kk}^{m_1\lambda m_2\lambda}.
\end{eqnarray}
They span irreducible representations:
\begin{eqnarray}
    \pi_{\mathcal G}^{(k)}&=&\xi_{\mathcal G}^{(k)}\otimes\id,\nonumber\\
    \pi_{\mathcal C}^{(k)}&=&\id\otimes\xi_{\mathcal C}^{(k)},
\end{eqnarray}
which are just the representations $\xi^{(k)}$ appearing in \eqref{eq:isodecomp} including their multiplicities when acting on $\mathbb C^n$. Indeed, we have \cite{Schlichtholz24}:
\begin{eqnarray}
    \label{Ut1}
    \pi_{\mathcal G}^{(k)}&=&\frac{1}{D^k_C}\sum_{m_1,m_2=1}^{D^k_G}\operatorname{Tr}\left(\pi_{\mathcal G}^{(k)}\hat\Pi^{m_1 m_2\dagger}_k\right)\hat\Pi^{m_1 m_2}_k\nonumber\\
   \pi_{\mathcal C}^{(k)}&=&\frac{1}{D^k_G}\sum_{\lambda_1,\lambda_2=1}^{D^k_C}\operatorname{Tr}\left(\pi_{\mathcal C}^{(k)}\hat\Lambda^{\lambda_1 \lambda_2\dagger}_k\right)\hat\Lambda^{\lambda_1 \lambda_2}_k.
\end{eqnarray}

The reason why we have to work with representations $\pi^{(k)}$ rather than $\xi^{(k)}$ is that these are the $\pi^{(k)}$ operators which can be operationally realized on the space $\mathbb C^n$, on condition that one has operational access to the Schur basis $\{\ket{k,m,\lambda}\}$. Since $\pi_{\mathcal C}$ is a commutant of $\pi_{\mathcal G}$, the Schur operator bases \eqref{PiBasis0} commute themselves:
\begin{equation}
    \label{PiCommute}
[\hat\Lambda^{\lambda_1\lambda_2}_k, \hat\Pi^{m_1 m_2}_k]=0.
\end{equation}

The most well known example of the above duality is the Schur-Weyl duality, defined by the fact that the representation $\pi_{\mathcal G}$ expresses a collective action of the group $\mathcal{G}$ on the tensor product space $ (\mathbb C^d)^{\otimes t}$:
\begin{equation}
    \label{colAction}
    \pi_{\mathcal G}(L)=L^{\otimes t},
\end{equation}
in which $L$ is an element of a defining $d$-dimensional representation of a group $\mathcal G$. Schur-Weyl duality is stated most often for group $\textrm{GL}(d, \mathbb C)$ and its subgroups $\textrm{SL}(d, \mathbb C)$
and $\textrm{U}(d, \mathbb C)$. In all those cases, the commutant $\pi_C$ turns out to be the symmetric group of a set of $t$ elements.

However, a less typical example with applications in quantum information is Schur-Weyl duality for the tensor power of a Clifford group \cite{Gross21}. Indeed, let us take $n$-qubit system with pure state space $(\mathbb{C}^2)^{\otimes n}$. The Pauli group is defined as a finite group generated by local application of Pauli unitaries $\{\sigma_1^{i_1}\otimes\ldots\otimes\sigma_n^{i_n}\}$. The Clifford group $Cl_n$ is defined as the normalizer of the Pauli group in the unitary group $\textrm{U}(2n)$, namely it is a set of unitaries $U_{Cl}$ from $\textrm{U}(2n)$, which map Pauli unitaries into Pauli unitaries. $Cl_n$, similarly to  the Pauli group, is a finite group. Now take $t$ copies of such an $n$-qubit system, described by the state space $((\mathbb{C}^2)^{\otimes n})^{\otimes t}\equiv (\mathbb C^2)^{\otimes nt}$. In analogy to the case of collective action of the unitary group on single quantum systems we can consider collective action of Clifford unitaries $U_{Cl}^{\otimes t}$ on $t$ copies of $n$-qubit systems.  The question arises what is the commutant of $\{U_{Cl}^{\otimes t}\}$ with respect to endomorphisms on $(\mathbb C^2)^{\otimes nt}$. In the work \cite{Gross21} such a commutant is explicitly constructed, which constructively defines Schur-Weyl duality for the collective action of Clifford unitaries. It is natural that the representations of permutations of $t$-element set on the space $(\mathbb C^2)^{\otimes nt}$ are present in the commutant, as they commute with all collective unitaries acting on this space. However, Clifford unitaries are a finite subgroup of all unitaries, hence the commutant is bigger and contains additional transformations, precisely specified in \cite{Gross21}.

Another very important example is the so-called mixed Schur-Weyl duality for the unitary groups \cite{MST_Nguyen23, Studzinski25}, in the case of which we consider the following representation of the unitary group: 
\begin{equation}
     \pi_{\operatorname{U}(d)}(U)=U^{\otimes t_1}\otimes\bar U^{\otimes t_2},
\end{equation}
on the tensor space $(\mathbb C^d)^{\otimes (t_1+t_2)}$, in which bar denotes complex conjugate of the matrix. The commutant of the matrix algebra generated by such representation is the so-called \textit{walled Brauer algebra} \cite{Studzinski25}.
Mixed Schur-Weyl duality has important applications, for example in the protocol of port-based teleportation \cite{Grinko23} and quantum error correction \cite{PRXQuantum.3.020314}.

Whenever one wants to utilize in practice the analytical methods related with duality of representations, one needs to implement the transition matrix from the computational basis in the representation space into the Schur basis \eqref{SchurBasis0}. Such transition matrix is called the \textit{Quantum Schur Transform} (QST) \cite{HarrowPHD}. 
So far there exist several efficient quantum circuit implementations of QST's for ordinary Schur-Weyl duality for unitary groups \cite{SWC_Bacon06, SWC_Kirby18, SWC_Krovi19} and for mixed Schur-Weyl duality for unitary groups \cite{MST_Nguyen23, Grinko23}.

\subsection{Main result -- universal finite averaging theorem}

Let us now formulate the main theorem of this work:
\begin{teor}
    Let $\mathcal G$ be any reductive Lie group, acting via a regular finite-dimensional representation $\pi_{\mathcal G}$, and let $\mathcal K$ and $\mathcal A$ be its Cartan components, maximally compact and maximally Abelian respectively. Then the Cartan-based average of arbitrary quantum state $\rho$, defined as:
    \begin{equation}
    \mathcal T_{\mathcal{G}}(\rho)=\int_{\mathcal K\times\mathcal A\times\mathcal K} \left(\pi_{\mathcal G}(K)\pi_{\mathcal G}(A_{\operatorname{n}})\pi_{\mathcal G}(K')\right)\rho\left(\pi_{\mathcal G}(K)\pi_{\mathcal G}(A_{\operatorname{n}})\pi_{\mathcal G}(K')\right)^{\dagger} \operatorname{d}\!K\operatorname{d}\!A \operatorname{d}\!K',
    \label{mainTwirlTeor}
\end{equation}
    where $K,K'\in{\mathcal K}$, $A_{\mathrm{n}}=A/||A||, \,A\in{\mathcal A}$, 
    $\operatorname{d}\!K$ and $\operatorname{d}\!K'$ represent Haar-invariant volume elements of $\mathcal K$ and $\mathcal K'$, whereas $\operatorname{d}\!A$ represents arbitrary normalized volume element of $\mathcal A$,
    can be always represented as a finite sum of the form:
    \begin{equation}
    \label{teor-fin-group-average}
    \mathcal T_{\mathcal G}(\rho)=\sum_{k\in\mathcal I} p_k \mathcal{T}_{\mathcal U}^k(\rho),
\end{equation}
in which:
\begin{equation}
    \label{teor-fin-group-average_mixed}
   \mathcal{T}_{\mathcal U}^k(\rho)=\frac{1}{\left(D^k_G\right)^2} \sum_{l=1}^{\left(D^k_G\right)^2} \tilde\gamma_l^{(k)}\rho\,\tilde\gamma_l^{(k)\dagger},
\end{equation}
is a $1$-design on $k$-th invariant subspace.
The operators $\{\tilde \gamma^{(k)}_l\}_{l=1}^{\left(D^k_G\right)^2}$ form a unitary operator basis acting on subspaces irreducible with respect to the action of $\pi_{\mathcal G}$:
\begin{equation}
    \label{Gt0}
\tilde\gamma^{(k)}_l=\sum_{m_1,m_2=1}^{D^k_G}{\left(\gamma^{(k)}_l\right)}_{m_1m_2}\hat\Pi^{m_1 m_2}_k,
\end{equation}
in which $\gamma^{(k)}_l$ denotes $l$-th unitary operator basis matrix acting on space $\mathbb C^{D^k_G}$, the operators $\hat\Pi^{m_1 m_2}_k$ represent Schur operator basis \eqref{PiBasis0} on these subspaces and $\tilde\gamma^{(k)}_l$ denotes its counterpart trivially extended as the operator on the entire representation space, on which $\rho$ is defined. The probabilities $p_k$ read:
$$p_k=\frac{\beta_k}{D^k},$$
in which the coefficient $\beta_k$ is specified by the following integral:
\begin{eqnarray}
\label{betafinG}
\beta_k=\Tr\left(\left[\int_{\mathcal A} \pi_{\mathcal G}(A_{\operatorname{n}})\pi_{\mathcal G}(A_{\operatorname{n}})^{\dagger}\operatorname{d}\!A\right]\hat\Pi_k^\dagger\right),
\end{eqnarray}
in which $\hat \Pi_k$ is the projector onto the entire $k$-th invariant subspace \eqref{SchurBasis0} and $D^k=D^k_CD^k_G$ denotes its dimension. In the case of $\mathcal G$ being compact, we have $p_k=1$ for each $k$.
\end{teor}
\begin{proof}
Due to the main result of the work \cite{Markiewicz23}, 
in the case of a Cartan-based averaging over a reductive symmetry group $\mathcal G$, with the additional assumption of uniform averaging over the compact components, the averaging procedure factorizes into a uniform averaging of the input state over the compact component $\mathcal{K}$ of the group and a constant (input-state-independent) factor corresponding to the uniform averaging \textit{of} the noncompact component $\mathcal A$ of $\mathcal G$ \textit{over} $\mathcal K$. This is illustrated by the following formula:
\begin{equation}
    \label{genRedTwirl00}
\mathcal{T}_{\mathcal{G}}(\rho) =\mathcal{T}_{\mathcal{K}}(\rho)\mathcal{T}_{\mathcal{K}}\left(\int_{\mathcal A}\pi_{\mathcal G}(A_{\operatorname{n}})\pi_{\mathcal G}(A_{\operatorname{n}})^{\dagger}\operatorname{d}\!A\right).
\end{equation}
In the work \cite{Markiewicz23} it is shown that the second factor has a very peculiar form: it is in fact rescaled identity operator, with constant factors on each of the $k$-th invariant subspaces:
\begin{equation}
    \label{genRedTwirl0}
\mathcal{T}_{\mathcal{K}}\left(\int_{\mathcal A}\pi_{\mathcal G}(A_{\operatorname{n}})\pi_{\mathcal G}(A_{\operatorname{n}})^{\dagger}\operatorname{d}\!A\right)=\sum_{k\in\mathcal I} \frac{\beta_k}{D^k}\hat\Pi_k,
\end{equation}
in which the coefficients $\beta_k$ are specified by the formula \eqref{betafinG}. Hence we have:
\begin{equation}
    \label{genRedTwirl}
\mathcal{T}_{\mathcal{G}}(\rho) =\mathcal{T}_{\mathcal{K}}(\rho)\sum_{k\in\mathcal I} \frac{\beta_k}{D^k}\hat\Pi_k.
\end{equation}
The second step boils down to showing that for arbitrary compact group $\mathcal K$  the uniform averaging over its unitary representation always can be represented as a finite averaging over unitary basis elements taken for each irreducible subspace. This holds because of Schur's Lemma as well as because of the fact that the Haar measure is left invariant 
with respect to shifts by group elements. Let us show this equivalence of averaging in details.

Firstly, let us decompose the map $\mathcal{T}_{\mathcal{K}}$, representing uniform averaging over a compact group $\mathcal K$, into irreducible representations:

\begin{equation}
    \label{block-diagonal-av}
    \mathcal T_{\mathcal K}(\rho)=\int_{\mathcal K}\operatorname{d}\!g\,\pi_{\mathcal K}(g)\rho\,\pi_{\mathcal K}(g)^{\dagger}=\sum_{k,k'\in\mathcal{I}} \int_{\mathcal K}\operatorname{d}\!g\,\pi_{\mathcal K}^{(k)}(g)\rho\,\pi_{\mathcal K}^{(k')}(g)^{\dagger}=\sum_{k\in\mathcal{I}} \int_{\mathcal K}\operatorname{d}\!g\,\pi_{\mathcal K}^{(k)}(g)\rho\,\pi_{\mathcal K}^{(k)}(g)^{\dagger}.
\end{equation}
This decomposition is diagonal in the index $k\in\mathcal{I}$ denoting irreducible representations as a consequence of Schur's lemma. A proof of this fact for $\mathcal K$ being a unitary group can be found in \cite{Markiewicz21}, Lemma 9. However it holds for uniform averaging over any compact group, since in such a case the Haar measure is invariant with respect to shift by a group element, which is the only essential assumption needed for proving the lemma.
Further we decompose the representation operators $ \pi_{\mathcal K}^{(k)}(g)$ in Schur operator bases \eqref{PiBasis0}, according to formulas \eqref{Ut1}, and the density matrix $\rho$ in the full Schur operator basis \eqref{FullPiBasis}:
\begin{eqnarray}
     \pi_{\mathcal K}^{(k)}(g)&=&\sum_{m_1,m_2=1}^{D^k_G}K(g)^k_{m_1m_2}\hat\Pi^{m_1 m_2}_k,\nonumber\\
\rho&=&\rho^{ij}_{n_1\lambda_1n_2\lambda_2}\hat\Pi_{ij}^{n_1\lambda_1n_2\lambda_2},
\end{eqnarray}
in which in the second line we use summation convention. The coefficients in the above formulas read:
\begin{eqnarray}
  K(g)^k_{m_1m_2}&=&\frac{1}{D^k_C}\Tr\left(\pi_{\mathcal K}^{(k)}(g)\hat\Pi^{m_1 m_2\dagger}_k\right),\nonumber\\
\rho^{ij}_{n_1\lambda_1n_2\lambda_2}&=&\Tr\left(\rho\hat\Pi_{ij}^{n_1\lambda_1n_2\lambda_2\dagger}\right).
\end{eqnarray}
Now, as shown in \cite{Markiewicz21}, Schur operator bases follow the following block-orthogonality condition, which follows solely from orthogonality of Schur basis:
\begin{equation}
\hat\Pi_k^{m_1m_2}\hat\Pi_{ij}^{n_1\lambda_1n_2\lambda_2}\hat\Pi_k^{r_1r_2}=\hat\Pi_k^{m_1m_2}\hat\Pi_k^{n_1n_2}\hat\Lambda_k^{\lambda_1\lambda_2}\hat\Pi_k^{r_1r_2}.
\end{equation}
Utilizing the above, as well as the commutativity relation \eqref{PiCommute}, we can express the twirling operation \eqref{block-diagonal-av} as follows (note that we use summation convention in each step):
 \begin{eqnarray}
    \label{unitTwirl3}
    \mathcal T_{\mathcal{K}}(\rho)&=&\int \left(K(g)^k_{m_1m_2}\hat\Pi^{m_1m_2}_k\right)\left(\rho^{ij}_{n_1\lambda_1n_2\lambda_2}\hat\Pi_{ij}^{n_1\lambda_1n_2\lambda_2}\right)\left((K(g)^\dagger)^k_{r_1r_2}\hat\Pi^{r_1 r_2}_k\right) \operatorname{d}\!g\nonumber\\
    &=&\int \left(K(g)^k_{m_1m_2}\hat\Pi^{m_1m_2}_k\right)\left(\rho^{kk}_{n_1\lambda_1n_2\lambda_2}\hat\Pi_{k}^{n_1n_2}\hat\Lambda_{k}^{\lambda_1\lambda_2}\right)\left((K(g)^\dagger)^k_{r_1r_2}\hat\Pi^{r_1 r_2}_k\right) \operatorname{d}\!g\nonumber\\
&=&\rho^{kk}_{n_1\lambda_1n_2\lambda_2}\hat\Lambda_{k}^{\lambda_1\lambda_2}\int \left(K(g)^k_{m_1m_2}\hat\Pi^{m_1 m_2}_k\right)\hat\Pi_{k}^{n_1n_2}\left((K(g)^\dagger)^k_{r_1r_2}\hat\Pi^{r_1 r_2}_k\right) \operatorname{d}\!g.
\end{eqnarray}
Now due to Schur's Lemma the integral in the above formula is proportional to identity on the entire $k$-th subspace:
\begin{equation}
\label{intOverGsubs}
    \int \left(K(g)^k_{m_1m_2}\hat\Pi^{m_1 m_2}_k\right)\hat\Pi_{k}^{n_1n_2}\left((K(g)^\dagger)^k_{r_1r_2}\hat\Pi^{r_1 r_2}_k\right) \operatorname{d}\!g=\frac{1}{D^k_G}\hat \Pi_k\delta_{n_1n_2},
\end{equation}
in which the projector $\hat\Pi_k$ is naturally defined as:
\begin{equation}
    \label{iProjML}
     \hat\Pi_k=\sum_{m=1}^{D^k_G}\hat\Pi_{k}^{mm}=\sum_{\lambda=1}^{D^k_C}\hat\Lambda_{k}^{\lambda\lambda}.
\end{equation}
The proof of this fact for unitary group can be found in \cite{Markiewicz21}, Appendix A.3, and is based on the following property of averaging over irreducible representations (see Lemma 10 therein):
\begin{equation}
\label{Schur1Twirl2}
    \int U_kX U^\dagger_k\operatorname{d}\!U
=\frac{1}{D_{\mathcal H}^k}\operatorname{Tr}(X)\id_{\mathcal H_k},
\end{equation}
in which $U_k$ is an irreducible representation of the unitary group on some complex vector space $\mathcal H_k$ of dimension $D_{\mathcal H}^k$ and $X\in \mathcal H_k$ is an arbitrary linear operator on $\mathcal H_k$.
However, since the proof does not utilize any properties of the unitary group apart from the fact that the Haar measure is invariant with respect to shifts by group elements, it holds for any compact groups acting via regular representations.

Now the crucial step is the fact, that the averaging \eqref{Schur1Twirl2} can be always replaced by a unitary $1$-design build up from a unitary operator basis \cite{BZ17}. Indeed let us take $\{\gamma^{(k)}_l\}$ as a unitary operator basis on $\mathcal H_k$, then we have:
\begin{equation}
\label{design1}
    \int U_kX U^\dagger_k\operatorname{d}U
=\frac{1}{D_{\mathcal H}^k}\operatorname{Tr}(X)\id_{\mathcal H_k}=\frac{1}{\left(D^k_{\mathcal H}\right)^2}\sum_{l=1}^{\left(D^k_{\mathcal H}\right)^2}\gamma^{(k)}_lX\gamma^{(k)\dagger}_l.
\end{equation}
Applying the above property to the integral \eqref{intOverGsubs} and noting that $\Tr\left(\hat\Pi_{k}^{n_1n_2}\right)=\delta_{n_1n_2}D^k_C$, we obtain:
\begin{eqnarray}
\label{intOverGsubs2}
    &&\int \left(K(g)^k_{m_1m_2}\hat\Pi^{m_1 m_2}_k\right)\hat\Pi_{k}^{n_1n_2}\left((K(g)^\dagger)^k_{r_1r_2}\hat\Pi^{r_1 r_2}_k\right) \operatorname{d}\!g\nonumber\\
    &&=\frac{1}{\left(D^k_G\right)^2}\sum_{l=1}^{\left(D^k_G\right)^2} \left(\left(\gamma^{(k)}_l\right)_{m_1m_2}\hat\Pi^{m_1 m_2}_k\right)\hat\Pi_{k}^{n_1n_2}\left(\left(\gamma^{(k)\dagger}_l\right)_{r_1r_2}\hat\Pi^{r_1 r_2}_k\right).
\end{eqnarray}
Inserting the last expression  into the last row of \eqref{unitTwirl3} and performing the transformations in \eqref{unitTwirl3} in reverse order starting from the last row, but with summation instead of integration we obtain:
\begin{equation}
     \mathcal T_{\mathcal{K}}(\rho)=\frac{1}{\left(D^k_G\right)^2}\sum_{l=1}^{\left(D^k_G\right)^2} \left(\left(\gamma^{(k)}_l\right)_{m_1m_2}\hat\Pi^{m_1 m_2}_k\right)\rho\left(\left(\gamma^{(k)\dagger}_l\right)_{r_1r_2}\hat\Pi^{r_1 r_2}_k\right)=\frac{1}{\left(D^k_G\right)^2}\sum_{l=1}^{\left(D^k_G\right)^2} \tilde\gamma^{(k)}_l\rho\,\tilde\gamma^{(k)\dagger}_l,
\end{equation}
in which in the last equality we used definition \eqref{Gt0}. Inserting the last expression into \eqref{genRedTwirl} we obtain the final result \eqref{teor-fin-group-average}.

\end{proof}

\section{Examples}

\subsection{Finite averaging over compact symmetry group}
Let us take $\mathcal G$ to be the unitary group $\textrm{U(2)}$ and let us take its representation on $(\mathbb C^2)^{\otimes 4}$ defined as a collective action: $\pi_{\mathcal{G}}(U)=U^{\otimes 4}$. Then the Cartan decomposition is trivial, and the twirling channel with respect to $\pi_{\mathcal{G}}$ \eqref{def-twirl-gen} represents averaging of a $4$-qubit quantum state over collective action of local unitary transformations. Let us see how such averaging can be represented as a finite averaging using unitary operator bases. In this case a standard Schur-Weyl duality holds and the action of $\pi_{\mathcal{G}}(U)=U^{\otimes 4}$ on the representation space 
$(\mathbb C^2)^{\otimes 4}$ decomposes into three  subspaces \eqref{SchurBasis0} of dimensions respectively $5$, $9$ and $2$ (including multiplicity), containing irreducible representations of  $\textrm{U(2)}$. Let us denote these subspaces by the index $k=1,2,3$. The $5$-dimensional subspace $k=1$ is spanned by  fully symmetric states, and corresponds to a single $5$-dimensional spin-$2$ irreducible representation
of $\textrm{U(2)}$. Let us denote the corresponding Schur basis vectors as $\{\ket{1, m, 1}\}_{m=1}^5$. The $9$-dimensional $k=2$ subspace comprises three equivalent $3$-dimensional spin-$1$ irreducible representations of $\textrm{U(2)}$, let us denote its Schur basis vectors as $\{\ket{2, m, \lambda}\}_{m,\lambda=1}^3$. The $2$-dimensional $k=3$ subspace  corresponds to two $1$-dimensional irreducible representations of $\textrm{U(2)}$, with corresponding Schur basis vectors $\{\ket{3, 1, \lambda}\}_{\lambda=1}^2$. Explicit representation of the described Schur basis in terms of standard basis can be found in Appendix A. Using \eqref{PiBasis0} we find 3 sets of Schur operator bases: $\{\hat\Pi^{m_1 m_2}_1\}_{m_1,m_2=1}^5$, $\{\hat\Pi^{m_1 m_2}_2\}_{m_1,m_2=1}^3$ and a single operator $\hat\Pi^{11}_3$. In order to define finite averaging \eqref{teor-fin-group-average} we need  to define two sets of unitary operator bases \eqref{Gt0} acting on spaces $\mathcal M(5, \mathbb C)$ ($5\times5$-dimensional complex matrices) and $\mathcal M(3, \mathbb C)$, since averaging on $1$-dimensional space is trivial.
In order to do so we choose Heisenberg-Weyl unitary operator bases \cite{BZ17}, comprising operators:
\begin{eqnarray}
\{\gamma^{(1)}_l\}_{l=1}^{25}&=&\{\omega_5^{ij}Z_5^iX_5^j\}_{i,j=0}^4,\nonumber\\    
\{\gamma^{(2)}_l\}_{l=1}^{9}&=&\{\omega_3^{ij}Z_3^iX_3^j\}_{i,j=0}^2,
\end{eqnarray}
in which the complex roots of unity read: $\omega_d=e^{2i\pi/d}$, and the generators $X_5, Z_5, X_3, Z_3$ are unitary generalizations of Pauli $\sigma_x$ and $\sigma_z$ operators in respective dimensions. Having all ingredients, and noting that due to compactness of the unitary group all the coefficients $p_k$ in \eqref{teor-fin-group-average} equal to one,  we can now formulate the final form of the finite averaging:
\begin{equation}
\label{fin:U2}
    \mathcal T_{\mathcal G}(\rho)= \frac{1}{25}  \sum_{l=1}^{25}\tilde\gamma^{(1)}_l\rho\tilde\gamma^{(1)\dagger}_l+\frac{1}{9}  \sum_{l=1}^{9}\tilde\gamma^{(2)}_l\rho\tilde\gamma^{(2)\dagger}_l+\hat\Pi^{11}_3\rho\hat\Pi^{11\dagger}_3.
\end{equation}
In the above, the operators $\tilde\gamma^{(1)}_l$ and $\tilde\gamma^{(2)}_l$ are defined by \eqref{Gt0}, whereas in the last term instead of writing $\tilde\gamma^{(3)}$ we use simply $\hat\Pi^{11}_3$, since in this $1$-dimensional case $\gamma^{(3)}=1$.

\subsection{Finite averaging over non-compact symmetry group}

Now, let us take $\mathcal G$ to be the group $\operatorname{SL}(2,\mathbb C)$, and consider $\pi_{\mathcal G}(L)=L^{\otimes 4}$, $L \in \textrm{SL}(2,\mathbb C)$. In order to define averaging procedure corresponding to this choice of group and its representation in accordance with the Cartan-based integration \eqref{mainTwirlG0}, let us first introduce Cartan 'KAK' decomposition for $\operatorname{SL}(2,\mathbb C)$. It reads: $L=SAS'$, in which $S$ and $S'$ are special unitary $\operatorname{SU}(2)$ operators, and $A$ is a real diagonal matrix representing the non-compact component of the group. Due to the commutativity of taking Cartan decomposition and applying a representation \eqref{diag:comm}, we can without loss of generality write:
$\pi_{\mathcal G}(L=SAS')=(SAS')^{\otimes 4}$. Now, the averaging procedure \eqref{mainTwirlG0}
represents averaging a four qubit system over a collective action of SLOCC-type operations \cite{Markiewicz21, Markiewicz23}. In order to express it as a finite averaging we need two elements: (i) find Schur basis related with the compact group component  in order to define Schur operator basis $\{\hat\Pi_k^{m_1m_2}\}$, (ii) find coefficients $\beta_k$ \eqref{betafinG}, related with non-compact part of the group $\textrm{SL}(2,\mathbb C)$. The Schur basis for the action of the group $\mathcal S=\textrm{SU}(2)$ via the representation $\pi_{\mathcal S}=S^{\otimes 4}$ is exactly the same as in the previous example for the action of the entire unitary group $\pi_{\mathcal G}=U^{\otimes 4}.$ Regarding the coefficients $\beta_k$ we can find them in an example shown in \cite{Markiewicz23}. Indeed, the normalized non-compact Cartan component of the group  $\textrm{SL}(2,\mathbb C)$ reads:
\begin{equation}
\label{AnSL2C}
     A_{\operatorname{n}} =  \left( \begin{array}{cc} 1 & 0 \\ 0 & x^{-2}  \end{array}  \right),\,\, x \geq 1 .
\end{equation}
Assuming that the above filtering operation is drawn according to the following Laguerre-like normalised measure:
\begin{equation}
    \label{supmeasure}
    \operatorname{d}\!A=\frac{e^{-x}dx}{\int_1^{\infty}e^{-x}dx},
\end{equation}
we obtain by numerical integration the following coefficients $\beta_k$ \cite{Markiewicz23}:
\begin{eqnarray}
\beta_1&\approx&0.30036,\nonumber\\
\beta_2&\approx&0.14652,\nonumber\\
\beta_3&\approx&0.12290.
\end{eqnarray}
Therefore following \eqref{teor-fin-group-average} we obtain a finite averaging over the group $\textrm{SL}(2,\mathbb C)$ as a rescaling of the corresponding expression for averaging over $\operatorname{U}(2)$ \eqref{fin:U2}:
\begin{equation}
    \mathcal T_{\mathcal G}(\rho)= \frac{\beta_1}{125}  \sum_{l=1}^{25}\tilde\gamma^{(1)}_l\rho\tilde\gamma^{(1)\dagger}_l+\frac{\beta_2}{81}  \sum_{l=1}^{9}\tilde\gamma^{(2)}_l\rho\tilde\gamma^{(2)\dagger}_l+\frac{\beta_3}{2}\hat\Pi^{11}_3\rho\hat\Pi^{11\dagger}_3.
\end{equation}

\section{Comparison with unitary $t$-designs and SL-$t$-designs}

Let us now focus on the case of collective action of symmetry groups on multiparticle quantum states specified by $\pi_{\textrm{U}}(U)=U^{\otimes t}$ for the case of unitary group $U\in\textrm{U}(d)$ and $\pi_{\textrm{SL}}(L)=L^{\otimes t}$ for the case of special linear group $L\in \textrm{SL}(d,\mathbb C)$. For such cases natural averaging sets have already been defined as unitary $t$-designs \cite{Dankert05} $\{U_i\}_{i\in\mathcal X_{\textrm{U}}}$ and only recently \cite{Markiewicz21} product  SL-$t$-designs of the form $\{L_i\}_{i=\{\alpha,\beta,\gamma\}\in\mathcal X_{\textrm{SL}}}$ with $L_i=U_{\alpha}A_{\beta}U_{\gamma}$. In the last construction $\{U_{\alpha}\}$ and $\{U_{\gamma}\}$ denote unitary $t$-designs, whereas $\{A_{\beta}\}$ denotes finite averaging set on the non-compact part of the group. In both cases the finite averaging is local with respect to the original physical subsystems:
\begin{eqnarray}
\label{localdesigns}
    \mathcal T_{\textrm{U}}(\rho)&=&\frac{1}{|\mathcal X_{\textrm{U}}|}\sum_{i\in\mathcal X_{\textrm{U}}}U_i^{\otimes t}\rho \,U_i^{\otimes t\dagger},\nonumber\\
    \mathcal T_{\textrm{SL}}(\rho)&=&\frac{1}{|\mathcal X_{\textrm{SL}}|}\sum_{i=\{\alpha,\beta,\gamma\}\in\mathcal X_{\textrm{SL}}}(U_{\alpha}A_{\beta}U_{\gamma})^{\otimes t}\rho \,(U_{\alpha}A_{\beta}U_{\gamma})^{\otimes t\dagger}.
\end{eqnarray}
On the other hand, our scheme of finite unitary averaging \eqref{teor-fin-group-average} is based on operations \eqref{teor-fin-group-average_mixed}, which are  not of the form of a tensor product of local operators. To see this  fact let us decompose both sampled operations comprising the finite averaging sets, $U_i^{\otimes t}$ and $ \tilde\gamma^{(k)}_l$ in the Schur operator basis \eqref{PiBasis0}:
\begin{eqnarray}
    \label{Ut1F}
   U_i^{\otimes t}&=&\sum_{k\in\mathcal I}\frac{1}{D^k_C}\sum_{m_1,m_2=1}^{D^k_G}\operatorname{Tr}\left( U_i^{\otimes t}\hat\Pi^{m_1 m_2\dagger}_k\right)\hat\Pi^{m_1 m_2}_k,\nonumber\\
    \tilde\gamma^{(k)}_l&=&\sum_{m_1,m_2=1}^{D^k_G}{\left(\gamma^{(k)}_l\right)}_{m_1m_2}\hat\Pi^{m_1 m_2}_k.
\end{eqnarray}
From the above decomposition it is clear that operators of the tensor product form translate into a direct sum of operators acting on each of the invariant subspaces carrying the irreducible representations of $\textrm{U}(d)$. Therefore, operators defined solely within particular subspaces, as is the case of $\tilde\gamma^{(k)}_l$, cannot be represented as tensor products of local operators: $\tilde\gamma^{(k)}_l\neq V^{\otimes t}$. This is because $V^{\otimes t}$, for any invertible $d\times d$ matrix $V$ decomposes in the Schur operator basis into subspaces corresponding to irreducible representations of $\textrm{GL}(d,\mathbb C)$, and due to their tensorial construction (see e.g. \cite{Tung}, chapter 13) they cannot be empty (consist of matrices of zeros) on any of these subspaces.

On the other hand, our construction is competitive when it comes to the size of finite averaging sets. Indeed, let us follow the notation of \cite{Roy09} and denote the dimension of the operator space spanned by  
tensor powers of unitary operators $U^{\otimes r}\otimes\bar{U}^{\otimes s}$, in which $U\in\textrm{U}(d)$ as $D(d, r, s)$. In the case of our construction the number of elements of a finite unitary averaging set is always equal to $D(d, t, 0)$, since it consists of a specific operator bases defined on all irreducible subspaces of the action of $U^{\otimes t}$. As proven in \cite{Roy09}, this number equals to:
\begin{equation}
    \label{Dto}
    D(d, t, 0)=\binom{d^2+t-1}{t}.
\end{equation}
At the same time, the \textit{lower bound} on the size of a unitary $t$-design, as proved in \cite{Roy09}, is specified by $D(d, \lceil t/2\rceil, \lfloor t/2\rfloor)$. Now, it happens that 
$D(d, t, 0)\leq D(d, \lceil t/2\rceil, \lfloor t/2\rfloor)$, and the equality appears only for $d=2$, namely for the case of $t$-qubit systems, see \cite{Roy09}, discussion after Theorem 8. Comparison of exemplary sizes of our finite averaging sets, lower bounds on sizes of the corresponding $t$-designs and sizes of known $t$-designs can be found in the Table \ref{tab:comp}. Our proposed unitary finite averaging sets are always at most as large as the lower bound on the size of unitary designs, however, for $d>2$ they always contain a smaller number of elements. For $t=2$ (averaging over $U^{\otimes 2}$) we can provide an exact scaling of these two approaches.
Namely, in the case of our approach we have the size of the averaging set equal to $D(d,2,0)=\binom{d^2+1}{2}=\frac{1}{2}d^4+\frac{1}{2}d^2$. In the case of unitary designs, the lower bound for the size of the unitary $2$-designs is known to be exactly $D(d,1,1)=d^4-2d^2+2$ \cite{Gross07}. As can be easily checked, for $d\geq 3$ we have:
$D(d,2,0)<D(d,1,1)$, and  asymptotically our construction contains half the number of elements of the lower bound for a size of a unitary $2$-design:
\begin{equation}
    \lim_{d\rightarrow\infty}\frac{D(d,2,0)}{D(d,1,1)}=\frac{\frac{1}{2}d^4+\frac{1}{2}d^2}{d^4-2d^2+2}=\frac{1}{2},
\end{equation}
which can be illustrated by comparing third and fourth columns of Table \ref{tab:comp}. 

In the case of SL-$t$-designs the discrepancy of the sizes of designs and our unitary finite averaging sets is even higher. For the latter, the number of elements is the same as in the case of averaging over tensor power of the unitary group, whereas  SL-$t$-designs are at least quadratically larger than the  unitary $t$-designs of the same dimensions.

\begin{center}
\begin{table}[]
\resizebox{\textwidth}{!}{
    \centering
    \begin{tabular}{||c | c | c | c | c | c ||} 
 \hline
 d & t & Univ. fin. set for $U^{\otimes t}$ and $L^{\otimes t}$ & Abs. bound for $U^{\otimes t}$ & Known design for $U^{\otimes t}$ & Known design for $L^{\otimes t}$\\ 
 \hline\hline
 2 & 2 & 10 & 10 & 12 & 1296\\ 
 \hline
 2 & 3 & 20 & 20 & 24 & 6336\\ 
 \hline
  2 & 5 & 56 & 56 & 60 & 54000\\ 
 \hline
  3 & 2 & 45 & 65 & 72 & ?\\ 
 \hline
  3 & 3 & 165 & 270 & 360 & ?\\ 
 \hline
  5 & 2 & 325 & 577 & 600 & ?\\ 
 \hline
  6 & 2 & 666 & 1226 & 2520 & ?\\ 
 \hline
  7 & 2 & 1225 & 2305 & 2352 & ?\\ 
 \hline
  8 & 2 & 2080 & 3970 & 20160 & ?\\ 
 \hline
  9 & 2 & 3321 & 6401 & 12960 & ?\\ 
 \hline
  10 & 2 & 5050 & 9802 & 95040 & ?\\ 
 \hline
\end{tabular}}
    \caption{Comparison of the sizes of finite averaging sets for averaging quantum states over tensor powers of unitary $U^{\otimes t}$ and special linear $L^{\otimes t}$ matrices for $U\in\textrm{U}(d)$ and $L\in\textrm{SL}(d,\mathbb C)$. The third column contains the sizes of the universal unitary finite averaging sets proposed in this work (they work for both, averaging over $U^{\otimes t}$ and  $L^{\otimes t}$), equal to $\binom{d^2+t-1}{t}$, see the main text. The fourth column contains the lower bound for the sizes of the corresponding unitary $t$-designs, equal to the dimension of the operator space $U^{\otimes \lceil t/2\rceil}\otimes\bar{U}^{\otimes\lfloor t/2\rfloor}$, see \cite{Roy09}. The fifth column contains sizes of the corresponding \textit{known} $t$-designs (after \cite{Roy09}), whereas the last column contains sizes of the so-called SL-$t$-designs (when known), as proposed in \cite{Markiewicz21}.}
    \label{tab:comp}
\end{table}
\end{center}

\section{Conclusions and Discussion}

In this work we have presented a construction of finite averaging sets for averaging finite dimensional quantum states over arbitrary representations of arbitrary matrix Lie groups. In the case of non-compact group the averaging procedure  is defined via iterated integration over Cartan-decomposed group, and the averaging is assumed to be uniform over the compact components and arbitrary over the non-compact ones. The proposed finite averaging sets consist of probabilistic mixtures of unitary $1$-designs defined on subspaces corresponding to irreducible representations of the group under consideration. The advantages of the construction are its universality and lower size of the averaging sets in the case of collective averaging when compared to respective $t$-designs. On the other hand, the main disadvantage is the fact that the finite averaging sets consist of operators which are not of the form of the representation matrices acting on the state, but of the form of operators acting on irreducible subspaces corresponding to these representations. Nevertheless, in the case of the most important applications, like averaging over tensor powers of the unitary operations $U^{\otimes t}$ and mixed unitary operations $U^{\otimes r}\otimes \bar{U}^{\otimes s}$, proposed construction can be efficiently implemented using several quantum circuit implementations of the so-called Quantum Schur Transforms.

Although the proposed procedure is quite general, it has two important constraints: we assume uniform averaging over the compact components of the symmetry group and we discuss only the case of finite dimensional representations. Natural generalization should go towards relaxing both constraints. 
In this way, one should try to construct universal finite unitary averaging with respect to arbitrary measures on the group elements. On the other hand, one may try to construct finite unitary averaging for physical systems described by infinite-dimensional state spaces, like relativistic quantum particles \cite{Schlichtholz24} or bosonic systems \cite{Serafini07}, using either the standard approach of this work on discrete subspaces corresponding to spin or helicity, by utilizing generalized quadratures or by following the direction of research in recently introduced \textit{rigged $t$-designs} for continuous variable systems \cite{rigged1, rigged2}.

\section*{Acknowledgements}
\noindent MM acknowledges support from the National Science Centre (NCN), Poland, under
Project Opus No. 2024/53/B/ST2/02026. 

\noindent This work is partially supported by the IRAP/MAB programme, project no. FENG.02.01-IP.05-0006/23, financed by the MAB FENG program 2021-2027, Priority FENG.02, Measure FENG.02.01., with the support of the FNP (Foundation for Polish Science).

\appendix
\section{Calculating Schur basis for the action of $\textrm{U}(2)$ on $ (\mathbb C^{2})^{\otimes 4}$}

In order to find Schur basis for the action of $\textrm{U}(2)$ on $ (\mathbb C^{2})^{\otimes 4}$ via $\pi_{\textrm{U}(2)}(U)=U^{\otimes 4}$ one can utilize the formalism of Young symmetrizers, see e.g. \cite{Tung} (chapter 5.5), or Appendix A of \cite{Markiewicz23}.
The three irreducible representations of $\textrm{U}(2)$, respectively spin-$2$ representation, spin-$1$ representation and spin-$0$ representation, that are realized on $ (\mathbb C^{2})^{\otimes 4}$ can be constructed using Young symmetrizers related with the following three Young diagrams:
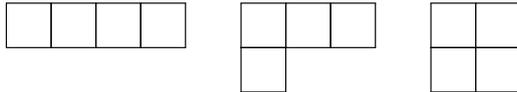
\begin{center}
    \begin{ytableau}
      \none   &  &  &  &
\end{ytableau}
\begin{ytableau}
      \none   &  &  &  \\
      \none   &
\end{ytableau}
\begin{ytableau}
      \none   &  &  \\
      \none   &  &
\end{ytableau}
\end{center}

\noindent For the fully symmetric $5$-dimensional subspace $\{\ket{1,m,1}\}_{m=1}^5$ carrying the spin-$2$ representation (first Young diagram) one finds the following basis states:
\begin{eqnarray}
\ket{1,1,1}&=&\ket{0000},\nonumber\\
    \ket{1,2,1}&=&\frac{1}{2}\left(\ket{1000}+\ket{0100}+\ket{0010}+\ket{0001}\right),\nonumber\\
    \ket{1,3,1}&=&\frac{1}{\sqrt{6}}\left(\ket{1100}+\ket{1010}+\ket{1001}+\ket{0110}+\ket{0101}+\ket{0011}\right),\nonumber\\
    \ket{1,4,1}&=&\frac{1}{2}\left(\ket{1110}+\ket{1101}+\ket{1011}+\ket{0111}\right),\nonumber\\
    \ket{1,5,1}&=&
\ket{1111}.\end{eqnarray}
For the $9$-dimensional subspace $\{\ket{2,m,\lambda}\}_{m,\lambda=1}^3$ carrying three equivalent spin-$1$ representations (the middle Young diagram) one gets the following non-orthogonalized basis elements:
\begin{eqnarray}
    e_{2,1,1}&=&\frac{1}{\sqrt{6}}(2\ket{0001}-\ket{1000}-\ket{0010}),\nonumber\\
    e_{2,1,2}&=&\frac{1}{\sqrt{6}}(2\ket{0010}-\ket{1000}-\ket{0001}),\nonumber\\
    e_{2,1,3}&=&\frac{1}{\sqrt{6}}(2\ket{0100}-\ket{1000}-\ket{0001}),\nonumber\\
    e_{2,2,1}&=&\frac{1}{\sqrt{6}}(2\ket{1110}-\ket{0111}-\ket{1101}),\nonumber\\
    e_{2,2,2}&=&\frac{1}{\sqrt{6}}(2\ket{1101}-\ket{0111}-\ket{1110}),\nonumber\\
    e_{2,2,3}&=&\frac{1}{\sqrt{6}}(2\ket{1011}-\ket{0111}-\ket{1110}),\nonumber\\
     e_{2,3,1}&=&\frac{1}{\sqrt{12}}(2\ket{0101}-2\ket{1010}+\ket{0011}-\ket{1100}+\ket{1001}-\ket{0110}),\nonumber\\
    e_{2,3,2}&=&\frac{1}{\sqrt{12}}(2\ket{0110}-2\ket{1001}+\ket{0011}-\ket{1100}+\ket{1010}-\ket{0101},\nonumber\\
    e_{2,3,3}&=&\frac{1}{\sqrt{12}}(2\ket{0110}-2\ket{1001}+\ket{0101}-\ket{1010}+\ket{1100}-\ket{0011}).
\end{eqnarray}
In order to obtain Schur basis vectors we need to orthogonalize them within subspaces $C^2_m$ for $m=1,2,3$ (see column subspaces in \eqref{SchurBasis0}):
\begin{equation}
    \{\ket{2,m,1}, \ket{2,m,2}, \ket{2,m,3}\}_{m=1}^3=\{\operatorname{Orthogonalize}(e_{2,m,1}, e_{2,m,2}, e_{2,m,3})\}_{m=1}^3.
\end{equation}
Finally, for the $2$-dimensional subspace $\{\ket{3,1,\lambda}\}_{\lambda=1,2}$ carrying two equivalent spin-$0$ representations (the last Young diagram), we get the following non-orthogonalized vectors:
\begin{eqnarray}
     f_{3,1,1}&=&\frac{1}{2}(\ket{1100}+\ket{0011}-\ket{1001}-\ket{0110}),\nonumber\\
      f_{3,1,2}&=&\frac{1}{2}(\ket{1010}+\ket{0101}-\ket{1001}-\ket{0110}).\nonumber\\
\end{eqnarray}
In order to obtain Schur basis vectors we also need to orthogonalize them within the $C^3_1$ subspace:
\begin{equation}
    \{\ket{3,1,1}, \ket{3,1,2}\}=\{\operatorname{Orthogonalize}(f_{3,1,1}, f_{3,1,2})\}.
\end{equation}

\bibliographystyle{quantum}

\end{document}